\documentclass[letterpaper, 10 pt, journal, twoside]{ieeetran}

\IEEEoverridecommandlockouts                              

\usepackage{amsmath}
\usepackage{mathtools}
\usepackage{amssymb}
\usepackage{mathabx}        
\usepackage{bbold}          
\usepackage{bm}
\usepackage{soul}
\usepackage{xcolor}
\let\labelindent\relax
\usepackage{enumitem}
\usepackage{url}
\usepackage{subfigure}
\usepackage{tabularx}
\usepackage{booktabs}
\usepackage[hidelinks]{hyperref}
\usepackage{xfrac}
\usepackage{tikz}
\usetikzlibrary{calc,angles,positioning,intersections,quotes,decorations.markings}

\newcommand{\rline}{{\mathbb R}}

\newcommand{\dist}{\text{dist}}
\newcommand{\bear}{\text{bearing}}

\newcommand{\kll}{{\frac{k_l -k_{l+1}}{2}}}

\newcommand{\bbm}[1]{\left[\begin{matrix} #1 \end{matrix}\right]}
\newcommand{\sbm}[1]{\left[\begin{smallmatrix} #1
   \end{smallmatrix}\right]}

\newcommand{\rfb}[1]{\mbox{\rm
   (\ref{#1})}\ifx\undefined\stillediting\else:\fbox{$#1$}\fi}

\newenvironment{proof}{\vspace{.1cm}\noindent{\sc Proof.}\hspace{0.10cm}\,\,}{$\hfill\Box$\vspace{.3cm}} 

\newtheorem{theorem}{Theorem}[section] 
\newtheorem{defi}[theorem]{Definition}

\newtheorem{lemma}[theorem]{Lemma} 
\newtheorem{pro}[theorem]{Proposition}

\title{\LARGE \bf
{Scaling up the formation of agents with heterogeneous sensing: mixed distance and bearing-only}
}

\author{Jin Chen$^{1}$ and Bayu Jayawardhana$^1$ and Hector Garcia de Marina$^2$
\thanks{*This work is supported by Center of Expertise Flexible Manufacturing Systems.}
\thanks{$^{1}$Jin Chen and Bayu Jayawardhana are with the Engineering and Technology Institute Groningen, Faculty of Science and Engineering, University of Groningen, the Netherlands. 
        {\tt\small jin.chen@rug.nl, b.jayawardhana@rug.nl}}%
\thanks{$^{2}$Hector Garcia de Marina is with Department of Computer Engineer, Automation and Robotics, and with CITIC, the University of Granada, 
Spain,  {\tt\small hgdemarina.@ugr.es}. The work of Hector Garcia de Marina is supported by the grant Ramon y Cajal RYC2020-030090-I from the Spanish Ministry of Science.}%
}

\begin{document}

\maketitle
\thispagestyle{empty}
\pagestyle{empty}

\begin{abstract}

Unlike the case with identical neighboring agents whose actions are mirrored, the problem of distributed formation control design with heterogeneous sensing is not straightforward. In this paper, we consider the problem of distributed formation control where each agent can only control distances or bearings with its neighbors. We firstly develop a rigidity theory with heterogeneous sensing to ensure that the desired shape is well-posed. Secondly, we propose an iterative method that allows us to scale up the number of agents with heterogeneous sensing. Finally, numerical simulations show the effectiveness of our method.

\end{abstract}

\section{Introduction}

Distributed formation control is one of the most actively studied topics in multi-agent systems \cite{oh2015survey}. Among its many potential applications, 
it is one of the key elements towards the yet unsolved problem of realizing reliable and predictable robot swarms \cite{yang2018grand}. The most popular design technique for distributed formation controllers is based on the  gradient descent of potential functions that seek to minimize the sum of inter-agent geometry errors, which can be defined based on 
distances  \cite{sun2014finite}, bearings \cite{zhao2015bearing}, or inner angles, \cite{chen2021maneuvering} in order to achieve  prescribed \textit{formation shapes} \cite{anderson2008rigid,dimarogonas2010stability,krick2009stabilisation,de2021distributed}.

In the literature on formation control, the use of homogeneous sensing mechanism for each agent  is ubiquitous. In this case, all agents maintain the same type of geometrical constraints, e.g., 
distance or bearing, with their neighboring agents \cite{oh2015survey}. Having two neighboring agents sharing the same formation controller is mathematically convenient for the stability analysis of the multi-agent systems, where the inherent symmetry can be exploited in the analysis. However, there is a lack of methodologies for scaling up formations where neighboring agents control different geometric variables. In this paper, we propose a systematic and iterative method to scale the number of agents while alternating the controlled geometric variable for the new agent that joins the formation network.

There are few works in literature that investigate such mixed constraints or heterogeneous sensing case. The paper \cite{bishop2015distributed} focuses on the combination of distance and bearing constraints that are defined in an undirected graph. Such an assumption entails that for a pair of neighboring agents, both agents control the same geometry. In this paper, we relax such an assumption, whereby for one pair of neighboring agents, one agent can control either the distance or the bearing between them. The authors in \cite{kwon2019generalized} consider the combination of distance, bearing and angle constraints at the same time, encompassing a general concept for \emph{rigidity}; still, as in \cite{bishop2015distributed} the neighboring agents have symmetrical control actions for the same edge in the graph. 

When neighboring agents control the same geometrical variable, it can restrict their application. Firstly, mirrored control actions typically demand neighboring robots to be equipped with the same type and quality of sensors. Definitively, different sensing mechanisms can break the symmetry execution expected by the control design. Secondly, certain robots may have neighboring robots that can only work with different geometrical constraints, e.g., distance constraint with one neighbor and bearing one with another neighbor. In this case, the robot must be equipped with multiple types of sensors. Thirdly, in unstructured environments, the sensing measurement may become less accurate in certain conditions. Let us provide two specific examples to illustrate the applicability of our approach. Firstly, we can use our method to check whether an agent can switch between sensing modes based on its initial distance to neighbors. For example, if our approach guarantees stability for both sensing modes at all times, then a merging agent can use bearing-based sensing when it is far from other agents and switch to distance-based sensing when it is close to its neighbors. Secondly, we can evaluate the impact of an agent on the robustness of formation by determining whether the desired shape remains stable regardless of the type of sensor used by that agent. If the formation shape remains stable regardless of the sensor type, then that agent can safely switch sensing mechanisms if its sensor systems fail. Otherwise, neighboring agents may need to switch their sensing mechanisms simultaneously to prevent instability in the formation shape. To guarantee stability across the entire network, we may need to evaluate all possible combinations of edges in the heterogeneous network. 

This paper addresses the formation control problem considering heterogeneous sensing between neighboring agents. More precisely, each neighboring agent controls inter-agent distances or bearings exclusively. Such problem has initially been  
studied in \cite{chan2021stability}, where the authors analyzed the stability of the two and three-agent cases rigorously.  
In particular, the authors in \cite{chan2021stability} consider each neighboring agent follows a local gradient-descent control action for its controlled variable, distance or bearing-only. However, when both local control actions are combined in the same graph's edge, then the stability problem becomes not trivial. It can even introduce robustness issues if the control gains and the desired geometry variables are not chosen carefully. 
Therefore, it introduces control design challenges for scaling up the number of agents in a formation control problem with heterogeneous neighboring agents. In this work, we explore and analyze the stability of merging heterogeneous neighboring agents into an existing (stable) heterogeneous formation. In particular, we provide a sufficient condition of the control gain for the new merging agent to guarantee the stability of the formation.  
Such results are based on the following two novel notions: (i). the heterogeneous sensing rigidity for defining the desired shape of a heterogeneous sensing formation; and (ii). the heterogeneous persistence for dealing with the interaction between a pair of heterogeneous neighboring agents. These conditions imply that scaling up the formation requires both the heterogeneous persistence and properly chosen control gains for the newly defined edges. 
Numerical simulations shows that either improper control gains or non-heterogeneous persistence can destroy the guarantee of reaching the desired formation shape.

This paper is organized as follows. We first present the required preliminaries in Section II. Then, we propose the theory of heterogeneous sensing rigidity in Section III to assist with defining the desired shape for a heterogeneous sensing formation. In the same Section III, we propose the definition of heterogeneous-sensing-persistent formation, which guarantees that our formation with \emph{unilateral interaction} topology has a feasible solution. Then, in the Section IV, we propose the formation control law dealing with neighboring agents with heterogeneous sensing, and we show how to merge newly heterogeneous agents to the existing formation. Numerical simulations are presented in Section V for illustrating merging case with six agents, and the possible consequences of not following our proposed design method. Finally, conclusions are presented in Section VI.

\section{Preliminaries and Problem Formulation}

\subsection{Notations}

Let $\left| \mathcal{S} \right|$ denote the cardinality of a given set $\mathcal{S}$. For a vector $x \in \mathbb{R}^n$, $x^{\top}$ denotes the transpose of $x$ and the 2-norm of $x$ is denoted by $\left\| x \right\| =\sqrt{x^{\top}x}$. The set of all combinations of linearly independent vectors $\upsilon _1,\dots ,\upsilon _k$ is denoted by $\mathbf{span}\left\{ \upsilon _1,\dots ,\upsilon _k \right\}$. For a vector $\upsilon\in\rline^2$, the notation $\angle \upsilon$ denotes the counter-clockwise angle of vector $\upsilon$ with respect to the unit vector $\sbm{1\\0}$, which lies on the $x$-axis of a global coordinate frame $\Sigma$. 
For a matrix $A\in \mathbb{R}^{m\times n}$, $\mathbf{Null}\left( A \right) \subset \mathbb{R}^n$, $\mathbf{trace}\left( A \right)$ and $\det \left( A \right) $ denote the null space, the trace and the determinant of $A$, respectively. The $n \times n$ identity matrix is denoted by $I_n$. For a given vector $\upsilon$, the notation $\mathrm{diag}\left( \upsilon \right)$ denotes the diagonal matrix whose diagonal elements are given by the entries of $ \upsilon$. Finally, for given matrices $A\in \mathbb{R}^{m\times n}$ and $B\in \mathbb{R}^{p\times q}$, the Kronecker product of $A$ and $B$ is denoted by $A\otimes B\in \mathbb{R}^{mp\times nq}$, and for a given dimension $d=2,3$, we denote $\tilde{A}=A\otimes I_d\in \mathbb{R}^{md\times nd}$ where $d=2$ refers to the $2D$ plane and $d=3$ corresponds to the $3D$ space.

\subsection{Graph theory}

A directed graph $\mathcal{G}$ is defined by the tuple $\left( \mathcal{V},\mathcal{E} \right)$, where $\mathcal{V}=\left\{ 1,2,\cdots ,n \right\} $ is the \textit{vertex} set and $\mathcal{E}\subseteq \mathcal{V}\times \mathcal{V}$ is the \textit{edge} set with $m=\left| \mathcal{E} \right|$ number of edges. 
If $(i,j)\in\mathcal{E}$, the ordered pair $(i,j)$ refers to the edge whose direction is represented by an arrow with node $i$ as the tail and node $j$ as the head. We assume throughout the paper that there is no self-loop in the considered graph $\mathcal{G}$, i.e., $\left( i,i \right) \notin \mathcal{E}$ for all $i\in \mathcal{V}$. 
The set of neighbors of vertex $i$ is denoted by $\mathcal{N}_i\triangleq \left\{ j\in \mathcal{V}|\left( i,j \right) \in \mathcal{E} \right\}$. Associated to $\mathcal{G}$, we define the incidence matrix $H\in \left\{ 0,\pm 1 \right\} ^{m\times n}$ with entries $\left[ H \right] _{ki}=-1$ if the vertex $i$ is the tail of edge $k$, $\left[ H \right] _{ki}=1$ if it is the head, and $\left[ H \right] _{ki}=0$ otherwise. Note that $\mathbf{span}\left\{ \mathbf{1}_n \right\} \subseteq \mathbf{Null}\left( H \right)$.

\subsection{Distance-based and bearing-only rigid formations}

{Given a finite collection of $n$ points $\left\{ p_i \right\} _{i=1}^{n}$ in $\mathbb{R}^d$ with $n\geqslant 2$ and $d\in \{2,3\}$, a configuration is denoted as $p=\left[ p_{1}^{\top},\dots ,p_{n}^{\top} \right] ^{\top}\in \mathbb{R}^{dn}$, where $p_i \in \mathbb{R}^{d}$ represents the position of agent $i$. Whenever it is clear from the context, the $i$-th robot is labeled as $Ri$, e.g. $R1$ refers to robot $1$. A framework in $\mathbb{R}^d$, denoted as $\mathcal{G}\left( p \right)$, is a combination of an undirected graph $\mathcal{G}\left( \mathcal{V},\mathcal{E} \right)$ and a configuration $p$, where vertex $i\in \mathcal{V}$ in the graph is mapped to the point $p_i$ in the configuration.

We assume $p_i\ne p_j$ if $i \ne j$, i.e., two agents can not be at the same position. For points $p_i \in \mathbb{R}^{2}$ and $p_j \in \mathbb{R}^{2}$ of the formation, we define the relative position vector as $z_{ij}=p_j-p_i\in \mathbb{R}^2$, the distance as $d_{ij}=\left\| z_{ij} \right\| \in \mathbb{R}_{>0}$ and the relative bearing as $g_{ij}=\frac{z_{ij}}{d_{ij}}\in \mathbb{R}^2$, all relative to a global coordinate frame $\varSigma ^g$. For illustrative purpose, these definitions are shown  
in Figure \ref{fig:preliminary}. By definition,  
$z_{ji}=-z_{ij}, d_{ji}=d_{ij}, g_{ji}=-g_{ij}$.

\begin{figure}
    \centering
    \includegraphics[width=0.5 \linewidth]{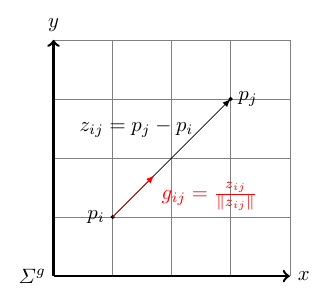}
    \caption{Illustration of the position of two agents $p_i$ and $p_j$ in $2$D plane relative to the global coordinate frame $\Sigma^g$. The relative position vector $z_{ij}$ and its bearing vector $g_{ij}$ are shown in black and red arrows, respectively.}
    \label{fig:preliminary}
\end{figure}

{An important methodology for studying and designing  
multi-agent formation control is the \textit{rigidity theory} that connects geometrical and graph properties of the desired formation shape. A number of different rigidity concepts have been proposed for the past decades based on the underlying information used to define the shape. Two widely studied ones are the  
\textit{distance-based rigidity theory} \cite{asimow1979rigidity} and the \textit{bearing-only based rigidity theory} \cite{zhao2015bearing}, \cite{trinh2018bearing}.} 

For the distance-based method, a \textit{rigid} formation is associated to a shape whose corresponding framework $\mathcal{G}\left( p \right)$ is \textit{distance rigid} \cite{anderson2008rigid, asimow1979rigidity}. 
As studied in \cite{anderson2008rigid, asimow1979rigidity}, an edge function $f_{\mathcal{G}}^{\text{dist}}\,\,: \mathbb{R}^{mn}\rightarrow \mathbb{R}^{\left| \mathcal{E} \right|}$ associated to the distance rigid framework \footnote{Later, we will introduce the edge function $f_{\mathcal{G}}^{\text{bearing}}$ to deal with bearings. In order to do not overload the notation, we will drop the superscript of $f_{\mathcal{G}}$ when it is clear from the context whether it is \emph{distance} or \emph{bearing}.} is given by
\begin{equation}\label{eq:edge_function_dist}
    f_{\mathcal{G}}^{\text{dist}}\left( p \right) =\underset{\forall k\in \left\{ 1,\cdots ,\left| \mathcal{E} \right| \right\}}{\text{col}}\left( \left\| z_k \right\| ^2 \right), 
\end{equation}
which collects all distance information in all edges of the (undirected\footnote{The usage of the index $k$ denotes the usage of the undirected version of the graph $\mathcal{G}$.}) graph $\mathcal G$. Correspondingly, a framework $\mathcal{G}\left( p \right)$ is {\it locally rigid} if for every $p\in \mathbb{R}^{mn}$, there exists a neighborhood $\mathcal{P}$ of $p$ such that $f_{\mathcal{G}}^{-1}\left( f_{\mathcal{G}}\left( p \right) \right) \cap \mathcal{P}=f_{\mathcal{K}}^{-1}\left( f_{\mathcal{K}}\left( p \right) \right) \cap \mathcal{P}
$, where $\mathcal{K}$ is a complete graph with the same vertex set $\mathcal{V}$ of $\mathcal{G}$. A framework is \textit{globally rigid} if $f_{\mathcal{G}}^{-1}\left( f_{\mathcal{G}}\left( p \right) \right) =f_{\mathcal{K}}^{-1}\left( f_{\mathcal{K}}\left( p \right) \right)$ for all $p\in \mathbb{R}^{mn}$. 

The change in the framework $\mathcal G(p)$ due to an infinitesimal displacement $\delta p$ of $p$ can be characterized by the Jacobian of $f_{\mathcal G}$: 
\begin{equation}
    \frac{\partial f_{\mathcal{G}}\left( p \right)}{\partial p}=2R\left( z \right) 
\end{equation}
where $R\left( z \right) \in \mathbb{R}^{\left| \mathcal{E} \right|\times m\left| \mathcal{V} \right|}$ is known as the {\it rigidity matrix} \cite{anderson2008rigid} of the framework $\mathcal{G}\left( p \right)$ with the undirected graph $\mathcal{G}$. The nontrivial kernel of $R(z)$ includes the translation and rotation of the whole framework in the Cartesian space. An {\it infinitesimal rigid} framework $\mathcal{G}\left( p \right)$ is a rigid framework that is invariant when is subjected to such infinitesimal displacement $\delta p$ (e.g. invariant under translation and rotation). In this case, we have 
$f_{\mathcal{G}}\left( p+\delta p \right) =f_{\mathcal{G}}\left( p \right)$ for all $p$. 
A framework $\mathcal{G}\left( p \right)$ is \textit{infinitesimally rigid} 
in $\mathbb{R}^2$ (resp. $\mathbb{R}^3$) if $\text{rank} R\left( z \right) =2n-3$ (resp. $3n-6$). A framework $\mathcal{G}\left( p \right)$ is \textit{minimally rigid} if it has exactly $2n-3$ (resp. $3n-6$) edges in $\mathbb{R}^2$ (resp. $\mathbb{R}^3$). 

Closely related to the distance-based rigidity framework described above, the bearing-only based rigid frameworks can be defined {\it vis-\`a-vis} as the distanced-based one where the edge function of distance rigid framework is replaced by bearing information, i.e. 
\begin{equation}\label{eq:edge_function_bearing}
    f_{\mathcal{G}}^{\text{bearing}}\left( p \right) =\underset{\forall k\in \left\{ 1,\cdots ,\left| \mathcal{E} \right| \right\}}{\text{col}}\left( g_k \right),
    \end{equation}
where $g_k=\frac{z_k}{d_k}$ is the relative bearing as defined before. 

}

\subsection{Distributed formation control law}

Using the framework $\mathcal G(p)$ as defined before, we assume that each agent position $p_i$ evolves in $\rline^d$ according to 
\begin{equation}\label{eq:agent_dynamics}
    \dot{p}_i=u_i, \quad i\in \left\{ 1,\dots ,n \right\}
\end{equation}
for all $i=1,\ldots, n$, where  
$u_i\in \mathbb{R}^d$ is the velocity control input. 
As before, we denote the concatenated control input by 
$u=\left[ \begin{matrix}	u_{1}^{\top}&		\cdots&		u_{n}^{\top}\\ \end{matrix} \right] ^{\top}\in \mathbb{R}^{nd}$.

\subsubsection{Distance-based formation control}

Using the distance rigidity framework, we can design a distance-based formation control in the following way. Let $d\in \rline^{|\mathcal E|}$ be the vector containing the desired distances associated with the set of desired relative positions $z^*$ such that $f_{\mathcal G}^{\text{dist}}(z^*)=d^*$, and define the distance formation error in the $\ell$-th edge $e_{\ell}^{\dist}=\|z_\ell\|^2-d_\ell^2$. Suppose that for every edge $\ell$, $V_{\ell}$ is a positive semi-definite function of $e_{\ell}^{\dist}$ such that $V^{\text{dist}}_\ell(e_{\ell}^{\dist})=0 \Leftrightarrow$ $e_{\ell}^{\dist}=0$, which implies that the edge distance is equal to the desired one.  
Accordingly, the corresponding gradient-based control law for each agent $i$ is given by 
\begin{equation*}
u_i^{\text{dist}} = -\sum_{\{\ell \, | \, (i,j)=\mathcal E_\ell^\dist\}} k_\ell \frac{\partial V^{\text{dist}}_\ell(e_{\ell}^{\dist})}{\partial p_i}^\top,
\end{equation*}
where $\mathcal E_\ell^\dist$ denotes the $\ell$-th edge in the set $\mathcal E^\dist$ and $k_\ell>0$ is an associated control gain. 
In this paper, we will study a particular distance-based potential function defined by $V^{\text{dist}}_{\ell}=\frac{k_l(e_{\ell}^{\dist})^2}{4 d_\ell^2}$, 
which gives us the following local control law  
\begin{equation}\label{eq:gradient_control}
u_i^{\text{dist}} = -\sum_{\{\ell \, | \, (i,j)=\mathcal E_\ell^\dist\}}k_\ell e_{\ell}^{\dist}\frac{z_{\ell}}{2d_\ell^2}.
\end{equation}
Note that the above control law is the distance-based formation control law which guarantees exponential stability due to the use of quadratic potential function $V^{\text{dist}}_{\ell}$ as also studied in 
\cite{sun2016exponential}.

\subsubsection{Bearing-only formation control}
Similar to the distance-based formation control above, using the bearing-only based rigid frameworks, the desired formation shape is defined by a set of inter-agent bearing constraints. Let $g^*\in\rline^{|\mathcal E|}$ define the desired bearing for all edges in $\mathcal E$ corresponding to the set of desired relative position $z^*$ such that $f_{\mathcal G}^{\text{bearing}}(z^*)=g^*$. 
Let us define the bearing error at the $k$-th edge by $e_{\ell}^\bear =g_{\ell} -g_{\ell}^{*}$. 

We will now define the corresponding potential function $V$ that can be used in the gradient control as in \eqref{eq:gradient_control}. For each edge $\ell$ with the associated agent pair $(i,j)$, we consider the following local potential function $V^{\text{bearing}}_{\ell}\left(z_{\ell},e_{\ell}^\bear\right) =k_l\left\| z_{\ell}^{*} \right\| \left\| z_{\ell} \right\| \left\| e_{\ell}^\bear \right\| ^2$. It is obvious that $V^{\text{bearing}}_{\ell} \geqslant 0$ and $V^{\text{bearing}}_{\ell}(e_{\ell}^\bear)=0\Longleftrightarrow e_{\ell}^\bear = 0\Longleftrightarrow  g_{\ell}  = g_{\ell}^{*}$. Despite $V^{\text{bearing}}_\ell$ has a dependence on both $z_\ell$ and $e_\ell^\bear$, the corresponding bearing-only gradient control law for an agent pair $(i,j)$ does not depend on $z_\ell$ and it is given by 

\begin{equation}\label{eq:bearing_gradient_control}
\begin{aligned}
    u^{\text{bearing}}_{i} & =-\sum_{\{l| (i,j)=\mathcal E_\ell^\bear \}} k_\ell\frac{\partial V^{\text{bearing}}_\ell(z_\ell,e_\ell^\bear)}{\partial p_i}^\top\\ 
   &= \sum_{\{l| (i,j)=\mathcal E_\ell^\bear \}}k_\ell\left\| z_{\ell}^{*} \right\|e_{\ell}^\bear 
\end{aligned}
\end{equation}
Note that this law is a modified version of the bearing-only formation control law (19) as presented in \cite{zhao2019bearing}. In comparison to the one used in \cite{zhao2019bearing}, we introduce a scaling factor to the potential function $V^{\text{bearing}}_{\ell}$ in each edge $l$, which does not change the sign property of the function and its time-derivative $\dot V^{\text{bearing}}_{\ell}$. Therefore, the modified control law still ensures the almost global asymptotic stability of the desired formation shape provided that the formation shape is infinitesimally bearing rigid. This scaling factor is introduced to ensure that the magnitude of the bearing-only and distance-based control inputs are of the same order.

{  
As discussed in the Introduction, the paper \cite{chan2021stability} explores distributed formation control with heterogeneous sensing mechanism containing two or three agents. It is concluded in \cite{chan2021stability} that when agents have different sensing mechanism on-board, there can exist some undesired formation shapes or group motions. For instance, when we consider the line formation of two agents with \textit{one distance and one bearing agent}, it can admit a moving formation where the two agents move together forming an incorrect line. It is further shown in \cite{chan2021stability} that such equilibrium point is unstable. The same observation is found in three agents case. In the formation of three agents with \textit{one distance and two bearing agents}, the agents may form an undesired formation shape with a constant velocity while that with \textit{two distance and one bearing agents}, the agents can form a flipped formation or a co-linear formation. 
}

\subsection{Problem Formulation}

{Based on the background information above, we can now formulate our distributed formation control problem where the neighboring agents use heterogeneous sensing information to maintain the formation. }
{As discussed in the Introduction, we study the formation system in which the  agents equip heterogeneous sensing devices so that each agent can only maintain the prescribed distance or bearing constraint with its neighbors using the available local distance or bearing information, respectively. Correspondingly, the desired shape of the formation is determined by a mixed set of distance and bearing sets and each agent use a  
different type of control law (e.g., distance-based or bearing-only control law) depending on the available local sensor systems. 

Let us consider two heterogeneous agent sets $\mathcal V^{\text{dist}}$ and $\mathcal V^{\text{bearing}}$. All agents $i$ in $\mathcal V^{\text{dist}}$ are distance-based agents which can only maintain desired distance to their neighbors given by the desired distance set $\bar{d}_i=\left\{ d_{ij}^{*}>0|d_{ij}^{*}\in \mathbb{R},j\in \mathcal{N}_i \right\}$. 
Each agent $i\in\mathcal V^{\text{dist}}$ is able to measure the relative positions with its neighbors in its \textit{local coordinate frame} $\varSigma^i$ which is not necessarily aligned with \textit{global coordinate frame} $\varSigma^g$. 
On the other hand, all agents $i$ in $\mathcal V^{\text{bearing}}$ are bearing-only agents whose goal are to maintain a desired bearing with their neighbors following the desired bearing set $\bar{g}_i=\left\{ g_{ij}^{*}>0|g_{ij}^{*}\in \mathbb{R},j\in \mathcal{N}_i \right\}$. Each agent $i\in\mathcal V^{\text{bearing}}$ is able to measure the relative bearings with its neighbors in its \textit{local coordinate frame} $\varSigma^i$ which is aligned with \textit{global coordinate frame} $\varSigma^g$. 

{\bf Distributed formation control with heterogeneous sensing problems: } Given the distance-based agents $\mathcal V^{\text{dist}}$ with the associated desired distance constraint set $\bar d_i$ and the bearing-only agents $\mathcal V^{\text{bearing}}$ with the associated desired bearing constraint set $\bar g_i$: 
\begin{description}
    \item[Q1.] Determine the geometric rigidity properties of the desired formation shape.
    \item[Q2.] For each agent $i$, design a local control law 
    \[
    u_i = f_i(z_\ell,e_\ell \, | \, \, \forall \ell \, \,  \text{ s.t. } (i,j)= \mathcal E_\ell)
    \]
    with Lipschitz continuous $f_i$ and with $e_\ell$ be either the distance error or bearing error on the $\ell$-th edge depending on the agent's type, such that all agents converge to the desired shape, i.e. $\lim_{t\to\infty} e_\ell(t) = 0$ for all $\ell$.
    \item[Q3.] Given a new agent $j$ to be added to the formation, determine the new formation configuration and local control law $u_j$ for agent $j$ and its neighboring agents such that all agents converge to the newly extended shape without compromising the stability of the global formation.
\end{description}    

This paper proposes to answer the three questions. Firstly, to address Q1, the paper introduces the concept of heterogeneous sensing rigidity in Section III. Secondly, the control law \eqref{eq:gradient_control} and \eqref{eq:bearing_gradient_control} in preliminaries already satisfy the problem Q2.} Finally, in Q3, we focus on the merging problem in heterogeneous sensing formation in which the new agent $\ell=n+1$ (either it is a distance or a bearing agent) will join a $n-$agent stable formation. This paper aims to investigate conditions so that the control laws that meet Q2 can handle the scaling up of the formation.  

\section{Heterogeneous Sensing} 

\subsection{Heterogeneous sensing rigidity}

Based on the aforementioned problem formulation, we propose in this section a rigidity graph framework that takes into account multiple sensing mechanism. The basic problem that it will address is whether a framework can be uniquely determined up to a translation factor given the relative bearings and distances between pairs of agents in the framework. 

Different from the framework in the distance or bearing-only based rigidity framework, the proposed heterogeneous sensing rigidity framework  
$\mathcal{G}\left( p \right)$ comprises two-layer directed graphs. 
Consider an arbitrary orientation of the graph $\mathcal{G}$ with the vertex set $\mathcal V$ and with two edge sets $\mathcal E^\dist$ and $\mathcal E^\bear$ associated to the distance-based edges and the bearing-only edges, respectively. Thus  $\mathcal{G}=\left(\mathcal{V},\mathcal{E}^{\dist}\cup\mathcal{E}^{\bear} \right)$. 
As before, by denoting $p$ as a configuration in $\mathcal G$, the framework  
$\mathcal{G}\left( p \right)$ is given by 
\begin{equation}
\begin{split}
    &z_{ij}\triangleq p_j-p_i,  \quad d_{ij}=\left\| z_{ij} \right\| , \quad \forall \left( i,j \right) \in \mathcal E^\dist
    \\
     &z_{ij}\triangleq p_j-p_i, \quad g_{ij}=\frac{z_{ij}}{\left\| z_{ij} \right\|}, \quad \forall \left( i,j \right) \in \mathcal{E}^{\bear} ,
\end{split}
\end{equation}
where the vectors $z_{ij}$ are the relative position based on $\mathcal E^\dist$ and $\mathcal{E}^{\bear}$, respectively, the scalar $d_{ij}$ is the distance between $p_j$ and $p_i$ in the graph $(\mathcal V,\mathcal E^\dist)$ and the unit vector $g_{ij}$ is the relative bearing of $p_j$ to $p_i$ in the graph $(\mathcal V,\mathcal{E}^{\bear})$. Then, following the standard notion to define rigidity by trivial motion, like \cite{zhao2015bearing}, \cite{sun2016exponential}, we try to introduce the concept of \textit{ infinitesimally heterogeneous sensing rigidity} in the same way.

\begin{defi}
The \textit{distance-bearing} function $F_{DB}\left( p \right) : \mathbb{R}^{dn}\rightarrow \mathbb{R}^{\left( d+1 \right) m}$ is defined by
\begin{equation}\label{eq:bearing-distance_func}
F_{DB}\left( p \right) \triangleq \left[ \left\| z_1 \right\| \,\,\cdots \,\,\left\| z_m \right\| \,\, g_{1}^{\top}\,\,\cdots \,\,g_{m}^{\top} \right] ^{\top}\in \mathbb{R}^{\left( d+1 \right) m}.
\end{equation}
\end{defi}

Analogous to the rigidity matrix and bearing rigidity matrix, we can define the \textit{distance-bearing} rigidity matrix as the Jacobian of the distance-bearing function
\begin{equation}
R\left( p \right) \triangleq \frac{\partial F_{DB}\left( p \right)}{\partial p}\in \mathbb{R}^{\left( d+1 \right) m\times dn}
\end{equation}
Let $\delta p$ denote an infinitesimal variation of the configuration $p$. If it is such that $R\left( p \right) \delta p=0$, then $\delta p$ is called an infinitesimal motion of $\mathcal{G}\left( p \right)$. Corresponding to the intersection of the infinitesimal motions of distance rigidity and bearing rigidity, the distance-bearing preserving motions of a framework $\mathcal{G}\left( p \right)$ are only the translational motions.

An infinitesimal motion is called trivial if it only corresponds to a translation of the entire framework. If all the infinitesimal motions are trivial in a framework $\mathcal{G}\left( p \right)$, $\mathcal{G}\left( p \right)$ is called {\em infinitesimally heterogeneous sensing rigid}. It is called {\em minimally infinitesimally heterogeneous sensing rigid} if it is infinitesimally rigid and has exactly $2n-2$ edges.

\subsection{Heterogeneous persistent formation} 

Since two different constraints are involved when the aforementioned heterogeneous sensing rigidity framework is used in the formation control, inconsistency issue may arise between these different constraints. As defined before, two edge sets $\mathcal E^\dist$ and $\mathcal{E}^{\bear}$ are used to embed the two different sensing information. Since we consider formation with heterogeneous sensing, there can be edges where the pairing agents use different sensing information so that each of them will maintain a different type of constraints; one of them uses distance constraint in its local control law and the other one employs bearing constraint in its local control law. In fact, the use of distance-bearing function $F_{DB}$ in \eqref{eq:bearing-distance_func} provides a formalism to embed such {\em heterogeneous constraint} coming from the two different edges $\mathcal{E}^\dist$ and $\mathcal{E}^{\bear}$. 

It is important to note that the notion of rigidity refers to an undirected graph, where two connected agents must maintain the same (distance or bearing-only) constraint. This notion can no longer be applicable when the two connected agents do not share the same constraint to maintain.  
  As an example, consider the scenario as shown in Figure 2(b). In this figure, agent $A_4$ can form an orbit around agent $A_2$ due to its distance constraint, and consequently it is not possible for agent $A_3$ to always maintain the given three distance constraints w.r.t. $A_1, A_2, A_4$. 
  In this case, a set of constraints $F_{DB}(p) = r^*$ for some desired vector of constraints $r^*\in\rline^{(d+1)m}$ does not admit a solution $p$.  
  In order to study such issues caused by the use of directed  graph in formation control, 
  the notion of persistent formation is introduced for the distance-based formation control [18], [19] and for the bearing-only formation [20]. Following the standard notion of persistent formation in these papers, we proceed to define heterogeneous sensing persistent formation as follows. At first, we will define the notion of 
{\em consistent heterogeneous constraint} $r^*$, which corresponds to the case when the desired constraint vector $r^*$ has an admissible solution $p^*$ to $F_{DB}^{-1}(r^*)=p^*$. 

Consider a two-layer directed graph $ \mathcal{G}=\left( \mathcal{V},\mathcal{E}^\dist \cup \mathcal{E}^{\bear} \right) $, a vector $ \bar{d}$ of desired distances $ d_{ij}^{*}>0$ with $ (i,j) \in \mathcal{E}^\dist$,  a vector $ \bar{g}$ of desired bearings $g_{ij}^{*}$ with $ (i,j) \in \mathcal{E}^{\bear}$, which is combined using $F_{DB}$ into a set of constraints $F_{DB}(p) = \sbm{\bar{d}\\ \bar{g}}$. A given configuration $p$ in a $2D$ plane needs to be checked whether it is compatible with the aforementioned set of constraints.  

Using $p$ and following the approach in \cite{yu2009control}, any edge $(i,j) \in \mathcal{E}^\dist$ is called {\it active} if $ \left\| p_i–p_j \right\| =d_{ij}^{*}$, i.e., the pair $p_i^*$ and $p_j^*$ satisfies the distance constraint assigned to the edge $(i,j)$. Similarly, any edge $(i,j) \in \mathcal{E}^{\bear}$ is called {\it active} if $ \frac{p_i-p_j}{\left\| p_i-p_j \right\|}=g_{ij}^{*}$, i.e., the pair $p_i$ and $p_j$ fulfills the bearing constraint assigned to the edge $(i,j)$.  Correspondingly, the position $p_i$ in the given configuration $p$ is said to be \textit{fitting position} if the number of active edges in $\mathcal{E}^\dist \cup \mathcal{E}^{\bear}$ involving $p_i$ cannot be increased by moving $p_i$ while keeping other vertices unchanged, e.g., there is no $\hat p_i\in \rline^2$ within a ball centered at $p_i$ such that 
\begin{align*}
& {\scriptstyle \{(i,j)\in\mathcal{E}^\dist|\|p_i-p_j\|=d_{ij}^*\}\cup\left\{(i,j)\in\mathcal{E}^{\bear}\, \Big|\, \frac{p_i-p_j}{\|p_i-p_j\|}=g_{ij}^*\right\}} \\ 
& {\scriptstyle \subset \{(i,j)\in\mathcal{E}^\dist|\|\hat p_i-p_j\|=d_{ij}^*\}\cup\left\{(i,j)\in\mathcal{E}^{\bear}\, \Big| \, \frac{ \hat p_i-p_j}{\|\hat p_i-p_j\|}=g_{ij}^*\right\}.} 
\end{align*} 
The configuration $p$ is called a {\it fitting configuration} of  $\mathcal{G}$ for $\bar{d}$ and $\bar{g}$ if all vertices $ \upsilon \in \mathcal{V}$ are fitting positions, i.e., $F_{DB}(p^*)=\sbm{\bar{d}\\ \bar{g}}$. 

A configuration $p$ is called \textit{heterogeneous constraint consistent} if there exists $ \epsilon >0 $ such that any fitting configuration $p^{\prime}$ with $ d\left( p,p^{\prime} \right) =\max_{i\in \mathcal{V}} \left\| p_i-p_{i}^{\prime} \right\| <\epsilon $ satisfies the distance constraint $\bar{d}$ and bearing constraint $\bar{g}$. The framework $\mathcal{G}(p)$ is called heterogeneous constraint consistent if  almost all of its configurations are heterogeneous constraint consistent. An example of a heterogeneous constraint consistent formation and non-consistent one are shown in Figure \ref{fig:consisten_and_nonconsistent}. 

{
\begin{figure*}
  \centering
  \subfigure[Heterogeneous constraint consistent]{\includegraphics[scale=0.6]{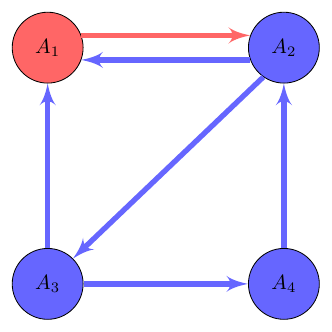}
  \label{fig:consisten_and_nonconsistent1}}
  \qquad 
  \subfigure[Heterogeneous constraint non-consistent]{\includegraphics[scale=0.6]{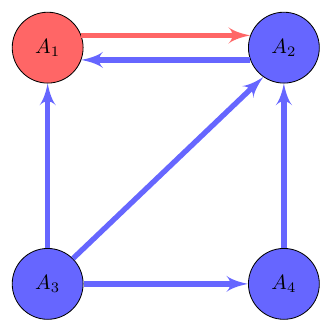}
  \label{fig:consisten_and_nonconsistent2}}
    \caption{The example of a heterogeneous constraint consistent and heterogeneous constraint non-consistent formation. In the non-consistent case (b), when the agents $1, 2$, and $3$ are each at the position which satisfies the geometrical constraints among themselves, agent $3$ is not able to simultaneously fulfill its third distance constraint w.r.t. agent $4$, as agent $4$ only needs to keep the required distance w.r.t. agent $2$. This situation does not occur in the consistent formation one in (a). }\label{fig:consisten_and_nonconsistent}
\end{figure*}
}

For the rest of the paper, we assume that the  desired formation satisfies both the heterogeneous rigid and heterogeneous constraint consistent conditions and we refer them simply as {\it heterogeneous persistent} formation.

\begin{defi}[heterogeneous persistent]
Consider a two-layer directed graph $\mathcal{G}=\left( \mathcal{V},\mathcal{E}^\dist \cup \mathcal{E}^{\bear} \right) $, a configuration $p$ of $\mathcal{G}$  
is called heterogeneous persistent if it is both infinitesimally heterogeneous sensing rigid and heterogeneous constraint consistent. The framework $\mathcal{G}(p)$ is called heterogeneous persistent if almost all of its configurations are heterogeneous persistent.
\end{defi}

\begin{defi}[minimally heterogeneous persistent]
A framework $\mathcal{G}\left( p \right)$ is minimally heterogeneous persistent, if it  
is both heterogeneous persistent and minimally infinitesimally heterogeneous sensing rigid.
\end{defi}

Theorem 2.3 in \cite{yu2009control} gives the relationship between the global Degrees of Freedom (DOFs) of the whole formation and the local DOFs of each agent. Based on this established relationship, we extend it to the heterogeneous sensing case. Correspondingly, it can be concluded that in a 2-D minimally heterogeneous persistent formation, each agent has to fulfill either none, one or two constraints (e.g., distance constraints or bearing constraints). The agent, which must fulfill two constraints, has zero DOF to move. Similarly, those with one constraint or zero constraint have one DOF or two DOFs,  respectively. To achieve a minimally heterogeneous persistent formation, it is necessary that the global DOFs of the whole formation are exactly two, representing two translation motions, which is equal to the sum of local DOFs of all agents in the formation. This can be achieved either with a single vertex having two DOFs or with two vertices, where each one has one DOF.

Consequently, depending on the underlying graph topology with the associated DOFs per agent, the heterogeneous persistent formation may contain one-leader or two-coleaders in a leader-follower setting, but not more than that\footnote{We refer interested readers to \cite{yu2009control} on a leader-follower setup in distributed formation control problems.}. Roughly speaking, one-leader formation refers to the formation framework in which there is only an agent (so-called the leader) that has two DOFs (in the desired formation) and all other agents have zero DOF. Similarly, two-coleaders formation corresponds to the formation framework where there are only two agents with one DOFs and the rest has zero DOF. 

In the subsequent section, we present the problem of extending the formation by adding a new agent to the formation framework. Depending on how the new agent interacts with some of the agents in the existing formation, we can classify the merging problem into two different cases: 
\begin{description}
\item[C1.] {\bf Unilateral-connection merging.} In this case, the agents in the original formation can be measured by the merging agent unilaterally and the original graph is extended by one vertex, where newly added edges all start from the added vertex.
\item[C2.] {\bf Interconnection merging.} In this case, both the merging agent and some of the existing agents can measure each other. This means that in the extended graph, there is an edge that starts from the new vertex and another edge that ends at the new vertex.  
\end{description}

As a consequence of the foregoing discussion, a new agent can only be added to the existing heterogeneous persistent formation while preserving this property if it is connected to the leader or one of the co-leaders of the original graph in the case of C2.

\section{Heterogeneous Sensing Formation Control and Merging}
In this subsection, we study an agent is added to an existing formation of $n$ heterogeneous sensing agents. In this case, the $n$ heterogeneous sensing agents are able to form a rigid formation with local stability and we investigate methods to allow another agent merging into such stable $n$-agent system and forming an asymptotically stable $n+1$ rigid formation shape. It extends the well-known results in homogeneous sensing formation case. If we consider the distance-based agents only or  bearing-only agents only, our methodology is the same as the standard homogeneous sensing formation control with directed graph as in \cite{yu2009control}, \cite{zhao2015bearing_persistence}. Depending on how the newly added edges are connected to the merging agent, we present three different analysis below, namely, (i). the case when it merges into a stable formation with two coleaders; (ii). the case when it merges into a stable formation with one leader; and (iii). the case when it merges unilaterally to a stable formation (i.e. case C1 as before).

In the following, we will employ linearization to analyze the stability. To facilitate the analysis, let us recall some key facts and introduce some lemmas that will be useful later. Firstly we denote a pair of unit vectors as $v_i = \sbm{\cos \alpha_i \\ \sin \alpha_i}$ and $\bar v_i = \sbm{\sin \alpha_i \\ -\cos \alpha_i}$. These vectors exhibit following relationship: $v_i^\top \bar v_i = \bar v_i^\top v_i =0$, $\bar v_i^\top v_j = -\bar v_j^\top v_i = \sin(\alpha_i - \alpha_j)$. For each $v_i$, we can define a projection matrix $P_i = v_iv_i^\top$. \vspace{0.1cm}

\begin{lemma}\label{lemma2}
    For a linear combination of projection matrices $D = k_iP_{i}+k_jP_{j}$ with positive constants $k_i$ and $k_j$, if $\left| \sin (\alpha_i - \alpha_j)\right| \neq 0$ then 
    $\frac{1}{k_i k_j \sin^2(\alpha_i - \alpha_j)} (\bar v_i\bar v_i^\top + \bar v_i\bar v_j^\top)=:D^{-1}$ is the inverse matrix of $D$.
\end{lemma}\vspace{0.1cm}

\begin{lemma}\label{lemma3}
    For a given complete graph $\mathcal{G} = \left( \mathcal{V},\mathcal{E} \right)$ and a linear combination of projection matrices $P = \sum_{i=1}^{n} k_iP_i$ with positive constants $k_i, i \in \{1,\cdots,n\}$, the eigenvalues of $P$ are given by
    \begin{equation}\nonumber
        \lambda_{1,2}={\scriptstyle \frac{\sum_{i=1}^{n} k_i}{2}\pm  \sqrt{\frac{\left( \sum_{i=1}^{n} k_i \right) ^2}{4}-\sum_{\{l \, | \, (i,j)=\mathcal E_l \}} k_ik_j\sin^2(\alpha_i - \alpha_j)}}
    \end{equation}
     where $l \in \{1, \cdots, \left| \mathcal{E}\right|\}$ and $\mathcal E_l$ represents edge $l$ in $\mathcal{G}$.
\end{lemma}\vspace{0.1cm}

Consider an autonomous system
\begin{equation}\label{system_lemma1}
    \begin{aligned}
        \dot x &= -Ax + \sbm{0 \quad \cdots \quad 0 \quad B}^\top y \\
        \dot y &= \sbm{0 \quad \cdots \quad 0 \quad C}x - Dy
    \end{aligned}
\end{equation}
where $x$ is in $\mathbb{R}^n$, y is in $\mathbb{R}^2$, $B,C,D \in \mathbb{R}^{2 \times 2}$ and are real symmetric matrices. In addition, $A+A^\top$ and $D$ are psitive definite.

\begin{lemma}\label{lemma4}
    System \eqref{system_lemma1} is exponential stable if $2\lambda_{\min}(A+A^\top) > \lambda^2_{\max}(B+C)\lambda_{\max}(D^{-1})$.
\end{lemma}

\vspace{0.2cm}

\begin{proof}
We choose $V = x^\top x + y^\top y$ as a Lyapunov function. let us define $x = \bbm{\bar x \\ \Tilde{x}}$ where $\bar x \in \mathbb R^{n-2}$ and $\Tilde x \in \mathbb R^{2}$. We know that
\begin{equation}
\begin{aligned}
    \dot V &= -x^\top (A+A^\top) x + \Tilde x^\top (B+C) y \\ & \quad + y^\top (B+C) \Tilde x -2 y^\top D y\top \\ &\leq -\lambda_{\min}(A+A^\top) \|\Tilde x\|^2 + \Tilde x^\top (B+C) y \\ & \quad + y^\top (B+C) \Tilde x -2 y^\top D y\top
\end{aligned}
\end{equation}
It shows that $\dot V < 0$ for any $\sbm{x \\ y} \neq 0$, if the matrix $M = \left[\begin{smallmatrix}
 \lambda_{\min}(A+A^\top)I_2&  -(B+C)& \\ 
 -(B+C)&  2D&\\ 
\end{smallmatrix}\right]$ is positive definite. Since $D$ is positive and the Schur component of $M$ is $S_{M} = \lambda_{\min}(A+A^\top)I_2 - \frac{1}{2}(B+C)D^{-1}(B+C)$, we conclude $M$ is positive definite if and only if $S_M$ is positive definite. Then we introduce the following inequality,
\begin{equation}
\begin{aligned}
    &z^\top S_M z = \lambda_{\min}(A+A^\top) \| z\|^2 - z^\top (\frac{1}{2}(B+C)D^{-1}(B+C))z \\
    &\geqslant \lambda_{\min}(A+A^\top) \| z\|^2 - \frac{1}{2} \lambda^2_{\max}(B+C)\lambda_{\max}(D^{-1})\| z\|^2 
\end{aligned}
\end{equation}
It implies that $\forall z \neq 0$, $z^\top S_M z>0$ if $2\lambda_{\min}(A+A^\top) > \lambda^2_{\max}(B+C)\lambda_{\max}(D^{-1})$. We prove the claim.
\end{proof}

\subsubsection{Interconnection merging into a stable formation with two coleaders}
\begin{figure}
    \centering
    \includegraphics[width=0.6 \linewidth]{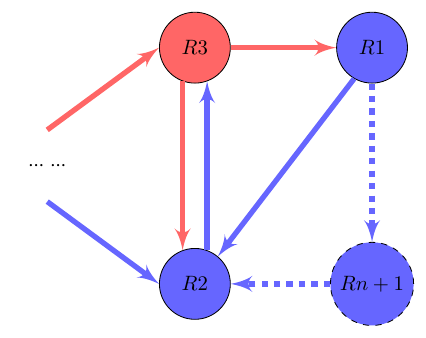}
    \caption{Interconnection merging into a stable formation with two coleaders; \protect\tikz{\protect\draw[fill = blue!60!white] (0, 0) circle [radius = 4pt];} represents a distance-based agent and \protect\tikz{\protect\draw[fill = red!60!white] (0, 0) circle [radius = 4pt];} represents a bearing-only agent; dashed circle \protect\tikz{\protect\draw[dashed, fill = blue!60!white] (0, 0) circle [radius = 4pt];} represents a merging distance-based agent. Correspondingly, blue dashed arrows are added sensing carried out by the distance-based agents.}
    \label{fig:two_leader}
\end{figure}

Let us first focus on the case where the new agent merge to the formation of $n$-agent systems by establishing edges to two co-leader agents that are connected by a link. The two agents in the original network and the newly added agent can be equipped by either distance-based or bearing-only sensing mechanism, and each of them uses correspondingly the distance-based or bearing-only gradient control law, respectively. As an example, we consider the $n$-agent systems as shown in Fig. \ref{fig:two_leader} where the robot $R(n+1)$ will merge to the formation by establishing two edges to the coleader robots $R1$ and $R2$ that are already connected by the link $R2\to R1$. 

When the new-added agent is distance-based robot, following the formation graph in Fig. \ref{fig:two_leader}, one can obtain that the closed-loop dynamics of $(n+1)$-agent systems is given by  
\begin{equation}\label{eq:Rn+1_distance_two_leaders}
    \left[ \begin{array}{c}
	\dot{p}_1\\
	\dot{p}_2\\
	\vdots\\
	\dot{p}_{n+1}\\
\end{array} \right] =\left[ \begin{array}{c}
	 k_1\frac{e_{1}^\dist z_{12}}{2\left\| z_{12}^{*} \right\|^2} + k_\ell\frac{e_{l}^\dist z_{1,n+1}}{2\left\| z_{1,n+1}^{*} \right\|^2}\\
	 k_2\frac{e_{2}^\dist z_{23}}{2\left\| z_{23}^{*} \right\|^2}\\
	 \vdots\\
	 k_{\ell+1}\frac{e_{l+1}^\dist z_{n+1,2}}{2\left\| z_{n+1,2}^{*} \right\|^2}\\
\end{array} \right] 
\end{equation}
\noindent where $k_1>0$, $k_2>0$, $k_\ell>0$ and $k_{\ell+1}>0$ are the controller gains. Similarly, when the $R(n+1)$ robot is a bearing-only one, the control input of $R(n+1)$ will be bearing-only in \eqref{eq:bearing_gradient_control}. One can obtain all other cases when the robots $R1$ and $R2$ use different sensing mechanism, and when the link direction is reversed. 

In the following proposition, we show that irrespective of the type of sensing mechanisms deployed in the two linked co-leader agents from the original network and the sensor systems deployed by the new agent $R(n+1)$, we can design the controller gains such that the enlarged network can maintain a new formation shape. Without loss of generality, we denote the two coleader agents in the original network by $R1$ and $R2$. 

\begin{pro}\label{pro:two_coleaders}
Consider the interconnection merging of robot $R(n+1)$ into a stable rigid formation of $n$ heterogeneous agents with the given gains $k_1,\ldots,k_{\ell+1}>0$ by adding two links from/to two co-leader agents $R1$ and $R2$ with desired distance $d^*_{1,n+1}$, $d^*_{2,n+1}$ and/or desired bearing $g^*_{1,n+1}$, $g^*_{2,n+1}$. Let $k_\ell$ and $k_{\ell+1}$ be the control gains of the local gradient-based control law with $k_\ell$ be the gain used by the agent from the original network and $k_{\ell+1}$ be the gain used by the newly added agent. For any given $k_l = k_{l+1} > 0$, there exists $k_\ell^* > 0$ such that for all $k_l = k_{l+1} > k_\ell^*$ the newly formed $n+1$ minimally heterogeneous persistent formation system achieves local exponential stability.
\end{pro}

\begin{proof}
In the following, we will prove the proposition for the case when $R(n+1)$ is a distance-based agent and the closed-loop system is given by \eqref{eq:Rn+1_distance_two_leaders}. For all other cases, the proof follows a similar fashion. 

In order to study the error dynamics, we analyse the dynamics of the edges that is given by the relative position vector $z$. For the original $n$-agent system, the closed-loop dynamics of relative position can be compactly described as 
\begin{equation} \label{relative_position_TL}
 \dot{z}=f\left( z \right) 
\end{equation}
where $ z=\left[ \begin{matrix}
	\cdots&		z_{ij}^{\top}&		\cdots\\
\end{matrix} \right] ^{\top}\in \mathbb{R}^{2\left( \left| \mathcal{V} \right|-1 \right)} $. Since by the hypothesis that the original $n$-agent forms an asymptotically stable rigid formation, the $z-$system above is locally asymptotically stable at an equilibrium point in $ \mathcal S=\left\{ z\in \mathbb{R}^{2\left( \left| \mathcal{V} \right|-1 \right)}|z=z^* \right\}$. This implies immediately that the linearization of \eqref{relative_position_TL} gives rise to 
\begin{equation}\label{z_hat_TL}
\dot{\hat{z}}=-A\hat{z}
\end{equation}
\noindent where $ \hat{z}=z-z^*$ and $-A$ is Hurwitz. On the rigidity property of the new configuration, the added two links to the new agent ensures that the enlarged formation remains infinitesimally heterogeneous sensing rigid. This follows from the Henneberg insertion rule for maintaining graph rigidity. 

Using \eqref{eq:Rn+1_distance_two_leaders}, we can analyze the relative position of $(n+1)$-agent system with state variables $z_m=\bbm{z^\top z_{1,n+1}^\top }^\top$, where $z$ is the relative position in the original formation framework.   

Because the new added robot $R(n+1)$ only influence the relative position $z_{21}$ in original $n$-agent system, we will focus on the interconnection that involves $\dot{z}_{12}$ and $\dot{z}_{1,n+1}$ in the analysis part below. Accordingly, we have the following dynamics of $z_m$ 
\begin{equation}\label{n+1_relative_position_TL}
\begin{aligned}
   \left[ \begin{array}{c}
	\vdots\\
	\dot{z}_{12}\\
	\dot{z}_{1,n+1}\\
\end{array} \right]
=\left[ \begin{array}{c}
	\vdots\\
	 k_2\frac{e_{2}^\dist z_{23}}{2\left\| z_{23}^{*} \right\|^2}-k_1\frac{e_{1}^\dist z_{12}}{2\left\| z_{12}^{*} \right\|^2} -k_\ell\frac{e_{l}^\dist z_{1,n+1}}{2\left\| z_{1,n+1}^{*} \right\|^2}\\
	 k_{\ell+1}\frac{e_{l+1}^\dist z_{n+1,2}}{2\left\| z_{n+1,2}^{*} \right\|^2}-k_1\frac{e_{1}^\dist z_{12}}{2\left\| z_{12}^{*} \right\|^2} -k_\ell\frac{e_{l}^\dist z_{1,n+1}}{2\left\| z_{1,n+1}^{*} \right\|^2}\\
\end{array} \right] 
\end{aligned}
\end{equation}
According to the edge relationship of triangle,  $z_{n+1,2}=z_{12}-z_{1,n+1}$. 
The linearization of the closed-loop dynamics \eqref{n+1_relative_position_TL} at $z_{m}^{*}$ yields
\begin{equation}\label{Linear_system_1}
\begin{aligned}
\dot{\hat{z}}_m &=\left[ \begin{array}{c}
	\dot{\hat{z}}\\
	\dot{\hat{z}}_{1,n+1}\\
\end{array} \right] 
= \left[ \begin{smallmatrix}
	-A&		 
	\begin{smallmatrix}
	0\\
	\vdots\\
	0\\
	B\\
\end{smallmatrix}\\
	\begin{smallmatrix}
	0&		\cdots&		0&		C\\
\end{smallmatrix}&		-D\\
\end{smallmatrix} \right] \left[ \begin{array}{c}
	\hat{z}\\
	\hat{z}_{1,n+1}\\
\end{array} \right],
\end{aligned}
\end{equation}
where $B = -k_\ell P_{l}$ , $C=k_{\ell+1}P_{l+1}-k_1P_{1}$ and $D= k_{\ell+1}P_{l+1}+k_\ell P_{l} $ with $P_{l}$, $P_{l+1}$, $P_{1}$ be the projection matrices as defined before Lemma \ref{lemma2} and $\alpha_1 = \angle z_{12}^{*}$, $\alpha_l = \angle z_{1,n+1}^{*}$, $\alpha_{l+1} = \angle z_{n+1,2}^{*}$. According to Lemma \ref{lemma4}, the system \eqref{Linear_system_1} is exponential stable if $2\lambda_{\min}(A+A^\top) > \lambda^2_{\max}(B+C)\lambda_{\max}(D^{-1})$. Based on Lemma \ref{lemma2}, $D^{-1}$ is given by
\begin{equation}
    D^{-1} = \frac{1}{k_lk_{l+1}\sin^2\theta_1}(\bar v_l \bar v_l^\top + \bar v_{l+1} \bar v_{l+1}^\top),
\end{equation}
where $\theta_1 = \alpha_l - \alpha_{l+1}$, $\bar v_l$ and $\bar v_{l+1}$ are defined before Lemma \ref{lemma2}. Additionally, Lemma \ref{lemma3} shows that
\begin{equation}
    \begin{aligned}
        &\lambda_{\max}(D^{-1}) =  \frac{1 + \sqrt{1 - \sin^2 \theta_1}}{k_lk_{l+1}\sin^2\theta_1} \\
        &\lambda_{\max}(B+C) = \frac{k_1+k_l-k_{l+1}}{2} + \\ & {\scriptstyle \sqrt{(\frac{k_1+k_l-k_{l+1}}{2})^2 +k_lk_{l+1}\sin^2 \theta_1 + k_1k_{l+1}\sin^2 \theta_2 - k_1k_l\sin^2 \theta_3}}
    \end{aligned}
\end{equation}
where $\theta_2 = \alpha_1 - \alpha_{l+1}$ and $\theta_3 = \alpha_1 - \alpha_l$. We know that $\sin\theta_1 \neq 0$, $\sin\theta_2 \neq 0$ and $\sin\theta_3 \neq 0$ because the desired positions of agents are noncollinear. By the hypothesis of the theorem, $k_l = k_{l+1}$. Accordingly,  $2\lambda_{\min}(A+A^\top) > \lambda^2_{\max}(B+C)\lambda_{\max}(D^{-1})$ can be given by
\begin{equation}
    2\lambda_{\min}(A+A^\top)(1 - \sqrt{1-\sin^2\theta_1})>f^2(\mu)
\end{equation}
where $\mu = \frac{k_1}{k_l}$ and $f(\mu) = \frac{\mu}{2} + \sqrt{(\frac{\mu}{2})^2 + \mu(\sin^2\theta_2 - \sin^2\theta_3) + \sin^2\theta_1}$. Since $\mu>0$, it follows that $f(\mu)$ is continuous, $f(\mu)>0$, $\lim_{\mu \rightarrow 0 } f(\mu) = \sin^2\theta_1$, $\lim_{\mu \rightarrow +\infty } f(\mu) = +\infty$ and $f^{'}(\mu) > 0$ when $\mu > 2(\sin^2\theta_3 - \sin^2\theta_2)$. Therefore, if 
\begin{equation}\label{ineq_A}
    \begin{aligned}
        &2\lambda_{\min}(A+A^\top)(1 - \sqrt{1-\sin^2\theta_1}) >f(0) \\
        &\Leftrightarrow \,\,\, 2\lambda_{\min}(A+A^\top) > 1 + \sqrt{1-\sin^2\theta_1}
    \end{aligned}
\end{equation}
we can conclude that there exists $\bar \mu >0$ such that  $2\lambda_{\min}(A+A^\top)(1 - \sqrt{1-\sin^2\theta_1}) >f(\mu)$ holds for all $\mu<\bar \mu$. We know the inequality \eqref{ineq_A} can be readily met by augmenting the gains within the existing system using the same multiplier. It shows that for all $k_l = k_{l+1} > \frac{k_1}{\bar \mu}$, we have that $2\lambda_{\min}(A+A^\top) > \lambda^2_{\max}(B+C)\lambda_{\max}(D^{-1})$. By the application of Lemma \ref{lemma4}, the augmented system \eqref{Linear_system_1} is exponentially stable. 

\end{proof}

\subsubsection{Interconnection merging into a stable formation with one leader}

\begin{figure}
    \centering
    \includegraphics[width=0.6 \linewidth]{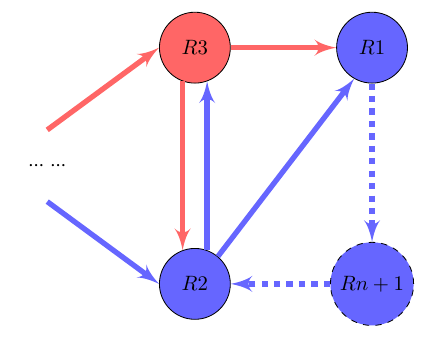}
    \caption{Interconnection merging into a stable formation with one leader}
    \label{fig:one_leader}
\end{figure}

Here we focus on the $n+1$ agents case in which $n$ heterogeneous sensing agents are able to form a one-leader minimally heterogeneous persistent formation with local stability while another agent is merging into such stable formation of $n$-agent systems. 
As an illustration, let us consider the case where two distance-based robots $R1$ and $R2$ are connected by a link with $R1$ be the leader of the formation. The merging robot $R(n+1)$ can be equipped with distance or bearing-only sensor system. Figure \ref{fig:one_leader} shows an example of the graph where the merging robot $R(n+1)$ can be detected by the leader $R1$ while the robot $R2$ can be detected by $R(n+1)$; hence they form an interconnected systems among themselves. 

When $R(n+1)$ is a distance-based robot, the closed-loop dynamics of $(n+1)$-agent system is given by 
\begin{equation}\label{n+1_one_leader_distance_based}
\left[ \begin{array}{c}
	\dot{p}_1\\
	\dot{p}_2\\
	\vdots\\
	\dot{p}_{n+1}\\
\end{array} \right] =\left[ \begin{array}{c}
	k_\ell \frac{e_{\ell}^\dist}{2\left\| z_{1,n+1}^{*} \right\|^2}z_{1,n+1}\\
	k_1\frac{e_{1}^\dist}{2\left\| z_{21}^{*} \right\|^2}z_{21}+k_2\frac{e_{2}^\dist}{2\left\| z_{23}^{*} \right\|^2}z_{23}\\
	\vdots\\
	k_{\ell+1}\frac{e_{\ell+1}^\dist}{2\left\| z_{n+1,2}^{*} \right\|^2}z_{n+1,2}\\
\end{array} \right], 
\end{equation}
where $k_1>0$, $k_2>0$, $k_\ell>0$ and $k_{\ell+1}>0$ are the controller gains. Similarly, if $R(n+1)$ is a bearing-only robot, the control input of $R(n+1)$ will be bearing-only in \eqref{eq:bearing_gradient_control}. For other cases when the robots $R1$ and $R2$ use different sensing mechanism, we can obtain similar expressions as above. 

In the following proposition, we show that irrespective of the type of sensing mechanisms deployed in the one leader agent from the original network and the sensor systems deployed by the new agent $R(n+1)$, we can design the controller gains such that the enlarged network can maintain a new formation shape. Without loss of generality, we denote the one leader agent in the original network by $R1$. 

\begin{pro}\label{pro:one_leader}
Consider the interconnection merging of robot $R(n+1)$ into a stable rigid formation of $n$ heterogeneous agents with the given gains $k_1,\ldots,k_{\ell-1}>0$ by adding one link from the leader agent $R1$ to $R(n+1)$ and another link from $R(n+1)$ to $R2$ with desired distance $d^*_{1,n+1}$, $d^*_{2,n+1}$ and/or desired bearing $g^*_{1,n+1}$, $g^*_{2,n+1}$. Let $k_\ell$ be the gain for the distributed control law on the link $R1 \to R(n+1)$ and $k_{\ell+1}$ be the gain for the distributed control law on the link $R(n+1)\to R2$. For any given $k_l = k_{\ell+1} > 0$, the newly formed $n+1$ minimally heterogeneous persistent formation system achieves local exponential stability.
\end{pro}

\begin{proof}
In the following, we will prove the proposition for the case when $R(n+1)$ is a distance-based agent and the closed-loop system is given by \eqref{n+1_one_leader_distance_based}. For all other cases, the proof follows vis-\`a-vis. 

Following the same process in the proof of Proposition \ref{pro:two_coleaders}, the linearization of the closed-loop relative position dynamics at desired relative positions yields
\begin{equation}\label{linear_n+1_one_leader}
\begin{aligned}
\dot{\hat{z}}_m &=\left[ \begin{array}{c}
	\dot{\hat{z}}\\
	\dot{\hat{z}}_{1,n+1}\\
\end{array} \right] 
= \left[ \begin{smallmatrix}
	-A&		 
	\begin{smallmatrix}
	0\\
	\vdots\\
	0\\
	B\\
\end{smallmatrix}\\
	\begin{smallmatrix}
	0&		\cdots&		0&		C\\
\end{smallmatrix}&		-D\\
\end{smallmatrix} \right] \left[ \begin{array}{c}
	\hat{z}\\
	\hat{z}_{1,n+1}\\
\end{array} \right],
\end{aligned}
\end{equation}
where $B=-k_\ell P_l$, $C=k_{\ell+1} P_{\ell+1}$ and $D=\left(k_\ell P_{\ell} +k_{\ell+1} P_{\ell+1}\right)$ with the projection matrices $P_{\ell}$, $P_{\ell+1}$ and $\alpha_l = \angle z_{1,n+1}^{*}$, $\alpha_{l+1} = \angle z_{n+1,2}^{*}$. In addition, $D$ and $A+A^\top$ are both positive definite. 

Note that following Lemma \ref{lemma4}, the system \eqref{linear_n+1_one_leader} is exponential stable if $2\lambda_{\min}(A+A^\top) > \lambda^2_{\max}(B+C)\lambda_{\max}(D^{-1})$. Following the same computation as in the proof of Proposition \ref{pro:two_coleaders}, we can arrive at the following sufficient condition 
\begin{equation}\label{inequ_2}
\begin{aligned}
    &2\lambda_{\min}(A+A^\top) > \\
    &{\scriptstyle \frac{\left(\kll + \sqrt{(\kll)^2 + k_lk_{l+1}\sin^2\theta}\right)^2}{k_lk_{l+1}\sin^2\theta}}\left( 1 + \sqrt{1 - \sin^2\theta} \right),
\end{aligned}
\end{equation}
where $\theta = \alpha_l - \alpha_{l+1}$. By assumption,  $\sin(\alpha_l - \alpha_{l+1}) \neq 0$ because the desired positions of agents are noncollinear. Since $k_l = k_{l+1}$ by the hypothesis of the proposition, the inequality \eqref{inequ_2} becomes
\begin{equation}
    2\lambda_{\min}(A+A^\top) > 1 + \sqrt{1 - \sin^2\theta}.
\end{equation}
The preceding inequality can be readily met by augmenting the gains within the existing system using the same multiplier. We prove the claim.

\end{proof}
   
\subsubsection{Unilateral-connection merging into a stable formation}

Finally, we focus on the case C2 on the unilateral-connection merging for $n+1$ agents. Similar to the setup in the previous case, it is assumed that there are $n$ heterogeneous sensing agents which are able to form a minimally heterogeneous persistent formation with local stability. In contrast to the previous case, the new robot $R(n+1)$ will merge to such stable formation by establishing two links unilaterally to any agents in the $n$-agent systems. 

Without loss of generality, we label the two agents where robot $R(n+1)$ establishes the new links by $R1$ and $R2$, respectively. Figure \ref{fig:unilateral_connection} shown an example of the graph where the merging robot $R(n+1)$ equipped with distance sensor system can detect $R1$ and $R2$. In this case, when $R(n+1)$ is a distance-based robot, 
the dynamics of the closed-loop system is given by
\begin{equation}\label{n+1_unilateral_distance}
\left[ \begin{array}{c}
	\vdots\\
	\dot{p}_{n+1}\\
\end{array} \right] =\left[ \begin{array}{c}
	\vdots\\
	k_\ell \frac{e_{\ell}^\dist}{2\left\| z_{n+1,1}^{*} \right\|^2}z_{n+1,1}+ k_{\ell+1} \frac{e_{\ell+1}^\dist}{2\left\| z_{n+1,2}^{*} \right\|^2}z_{n+1,2}\\
\end{array} \right], 
\end{equation}
where $k_\ell>0$ and $k_{\ell+1}>0$ are the controller gains. Similarly, if $R(n+1)$ is a bearing-only robot, the control input of $R(n+1)$ will be bearing-only in \eqref{eq:bearing_gradient_control}.

The following proposition shows that independent of the type of sensing mechanism in the three robots $R1$, $R2$ and $R(n+1)$, we can design the controller gains $k_\ell$ and $k_{\ell+1}$ such that the new formation is stable.  

\begin{pro}\label{pro:unilateral_case}
   Consider the unilateral connection merging of robot $R(n+1)$ into a stable rigid formation of $n$ heterogeneous agents with the given gains $k_1,\ldots,k_{\ell-1}>0$ by adding two links to the agents $R1$ and $R2$, which are not collinear, with desired distance $d^*_{n+1,1}$, $d^*_{n+1,2}$ or desired bearing $g^*_{n+1,1}$, $g^*_{n+1,2}$. Let $k_\ell$, $k_{\ell+1}$ be the gains for the distributed control laws on the links $R(n+1)\to R1$ and $R(n+1)\to R2$, respectively. For any given $k_{\ell}, k_{\ell+1}>0$, the newly formed $n+1$ minimally heterogeneous persistent formation system achieves local exponential stability.
\end{pro}
\begin{figure}
    \centering
    \includegraphics[width=0.6 \linewidth]{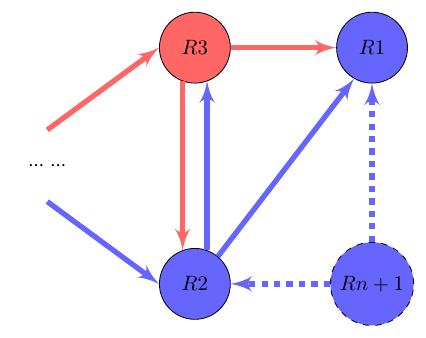}
     \caption{Unilateral-connection merging into a stable formation}
    \label{fig:unilateral_connection}
\end{figure}

\begin{proof}
In the following, we will prove the proposition for the case when $R(n+1)$ is a distance-based agent and the closed-loop system is given by \eqref{n+1_unilateral_distance}. For the case when $R(n+1)$ is a bearing-only agent, the proof follows similarly. 

Following the same process in the proof of Proposition \ref{pro:two_coleaders}, the linearization of the closed-loop relative-position dynamics at desired relative positions yields 
\begin{equation}\label{z_m_hat_unilateral}
\begin{aligned}
\dot{\hat{z}}_m=\left[ \begin{array}{c}
	\dot{\hat{z}}\\
	\dot{\hat{z}}_{n+1,1}\\
\end{array} \right] =\left[ \begin{matrix}
	-A&		0\\
	*&		-D\\
\end{matrix} \right] \left[ \begin{array}{c}
	\hat{z}\\
	\hat{z}_{n+1,1}\\
\end{array} \right]
\end{aligned}
\end{equation}
where $B=-k_\ell P_l$, $C=k_{\ell+1} P_{\ell+1}$ and $D=\left(k_\ell P_{\ell} +k_{\ell+1} P_{\ell+1}\right)$ with $P_{\ell}$, $P_{\ell+1}$ be the projection matrices as defined before Lemma \ref{lemma2} and $\alpha_l = \angle z_{1,n+1}^{*}$, $\alpha_{l+1} = \angle z_{n+1,2}^{*}$. Because $-A$ is Hurwitz and $-D$ are negative definite, we can conclude the system \eqref{z_m_hat_unilateral} is globally exponential stable. This implies that  \eqref{n+1_unilateral_distance} is locally asymptotically stable at desired relative positions. 

\end{proof}

\section{SIMULATION}
\subsection{Simulation setup}
In this subsection, we present simulation setup for 6-agent heterogeneous sensing formation that is formed by merging two extra agent at-a-time. The merging process incorporates all the different merging cases described in the previous section, where we start from two agent systems and it is expanded to 6-agent system. The different merging cases in this setup are summarized in Figure \ref{fig:simulation_graph}. For the desired formation, we set the distance and bearing constraints as $\{ d_{12}^* = 5, g_{21}^* = \bbm{1, 0}^\top, d_{13}^*=5\sqrt{2}, g_{23}^* = \bbm{0, -1}^\top, g_{34}^* = \bbm{-\frac{\sqrt{2}}{2}, \frac{\sqrt{2}}{2}}^\top, d_{42}^*=5, d_{45}^*=7.66, g_{56}^* = \bbm{1 , 0 }^\top, d_{61}^*= 7.66, d_{62}^*= 9.05\}$.

\begin{figure}[!t]
    \centering
    {
    \subfigure[First merging case: a bearing-only agent $R3$  and a distance-based agent $R4$ join the $2$-agent formation systems]
    {
    \includegraphics[width=0.25\textwidth]{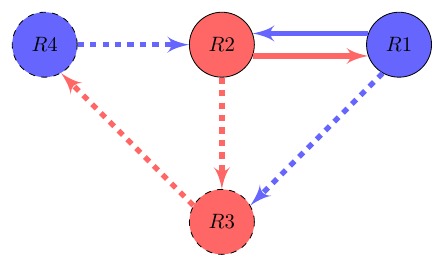} 
    \label{fig:simulation_graph_4_agent}
    }\,\,\, 
    \subfigure[Second merging case: a bearing-only agent $R6$ and a distance-based agent $R5$ join the $4$-agent formation systems]
    {\includegraphics[width=0.25\textwidth]{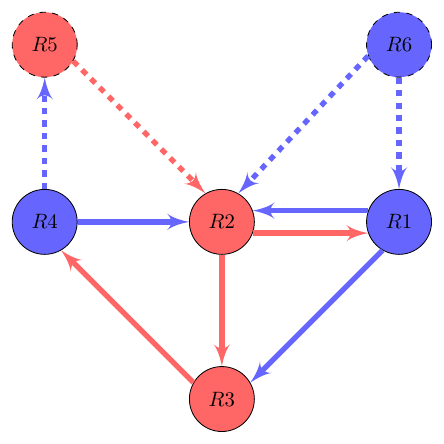} 
    \label{fig:simulation_graph_6_agent}
    }
    }
    \caption{Underlying graphs of formation system in simulation 
    }
    \label{fig:simulation_graph}
\end{figure}

\subsection{Simulation results}

In this part, we present simulation results of various the merging processes from 2-agent heterogeneous formation towards 6-agent heterogeneous  formation as described before.

After scaling up, the closed-loop dynamics is given by \begin{equation}
    \left[ \begin{array}{c}
	\dot{p}_1\\
	\dot{p}_2\\
	\dot{p}_3\\
	\dot{p}_4\\
	\dot{p}_5\\
        \dot{p}_6
\end{array} \right] = \left[ \begin{array}{c}
	k_1\frac{e_{1}^{\dist}z_{12}}{2\left\| z_{12}^{*} \right\| ^2}+k_2\frac{e_{2}^{\dist}z_{13}}{2\left\| z_{13}^{*} \right\| ^2}\\
	k_3\left\| z_{21}^{*} \right\| e_{1}^{\bear} + k_4\left\| z_{23}^{*} \right\| e_{2}^{\bear}\\
	k_5\left\| z_{34}^{*} \right\| e_{3}^{\bear}\\
	k_6\frac{e_{3}^{\dist}z_{42}}{2\left\| z_{42}^{*} \right\| ^2} + k_7\frac{e_{4}^{\dist}z_{45}}{2\left\| z_{45}^{*} \right\| ^2}\\
	k_8\left\| z_{56}^{*} \right\| e_{4}^{\bear}\\
 k_9\frac{e_{5}^{\dist}z_{61}}{2\left\| z_{61}^{*} \right\| ^2}+k_{10}\frac{e_{6}^{\dist}z_{62}}{2\left\| z_{62}^{*} \right\| ^2}
\end{array} \right] 
\end{equation}
where $k_1=k_3=3$ as control gains of original two agents. Following the computations in the proof of Proposition \ref{pro:two_coleaders}, the lower bound for $k_2$ and $k_4$ are given by $k_2^*=\frac{12\sqrt{3}}{23}$. For our selection, we opt for $k_2= k_4 =3$. In accordance with Proposition \ref{pro:one_leader}, we can choose any $k_5=k_6>0$. To ensure a similar convergence rate, we set $k_5=k_6 = 3$. Employing the same computations, we derive the lower bound of $k_7$ and $k_8$ as $3.695$ and determine that $k_9$ and $k_{10}$ can be any positive number. Accordingly, we set $k_7=k_8 = 3.8$ and $k_9 = k_{10} =3$. In the simulation, all gains are set as mentioned, and the result is shown in Figure \ref{fig:simulation_result_6_agent}. It can be seen from this figure that the six-agent formation system is able to maintain the desired formation shape. 
\begin{figure}
  \centering
  {
  \subfigure[Original 2-agent system]{\includegraphics[width=0.48\columnwidth]{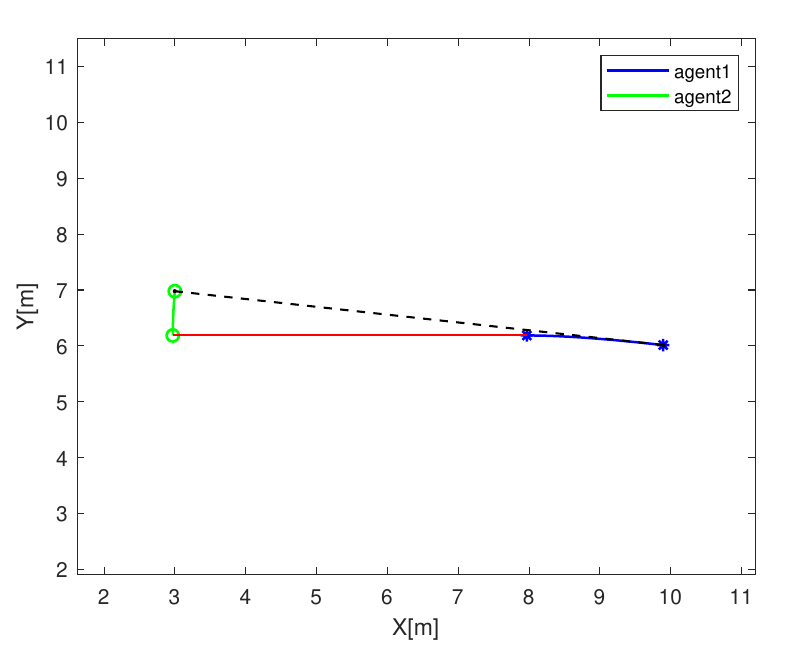}
  \label{fig:simulation_result_2_agent}}
  \subfigure[4-agent system]{\includegraphics[width=0.48\columnwidth]{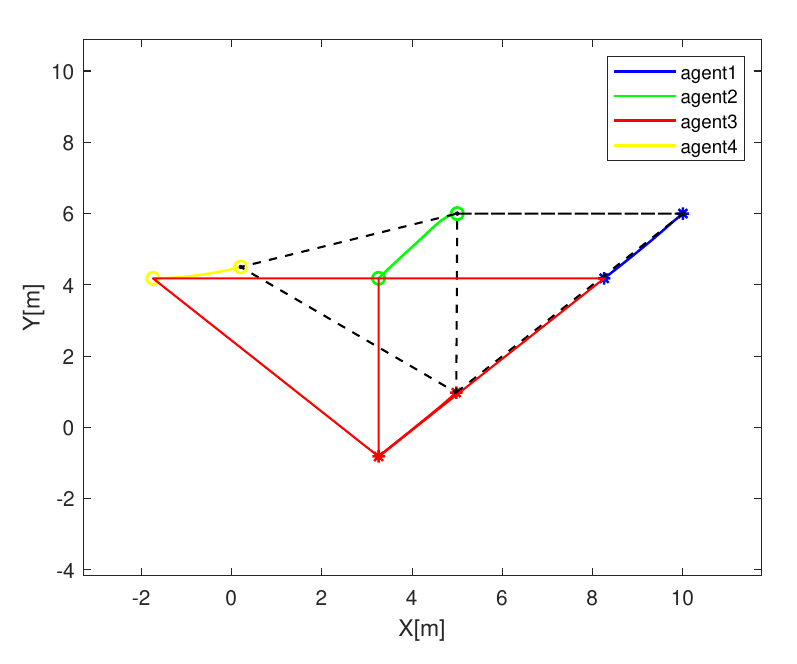}
  \label{fig:simulation_result_4_agent}}
  \subfigure[6-agent system]{\includegraphics[width=0.85\columnwidth]{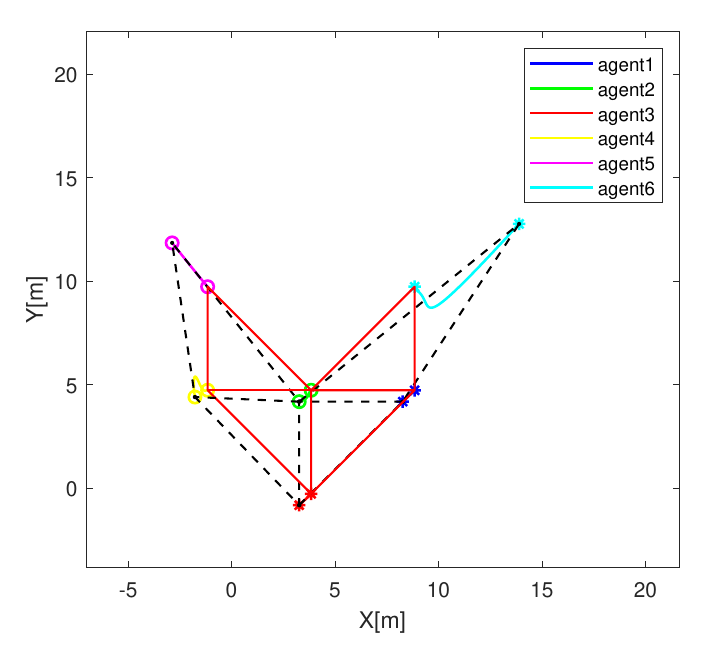}
  \label{fig:simulation_result_6_agent}}
  }
  \caption{Robot trajectories for the 6-agent heterogeneous sensing formation system setup; 
    (\protect\tikz{\protect\draw[blue, line width = 3] (0, 0) -- (0.2, 0);},
    \protect\tikz{\protect\draw[green, line width = 3] (0, 0) -- (0.2, 0);},
    \protect\tikz{\protect\draw[red, line width = 3] (0, 0) -- (0.2, 0);},
    \protect\tikz{\protect\draw[yellow, line width = 3] (0, 0) -- (0.2, 0);},
    \protect\tikz{\protect\draw[magenta, line width = 3] (0, 0) -- (0.2, 0);},
    \protect\tikz{\protect\draw[cyan, line width = 3] (0, 0) -- (0.2, 0);}) 
    = (\texttt{R1}, \texttt{R2}, \texttt{R3},\texttt{R4},\texttt{R5},\texttt{R6}), 
    $ \circ $ represents the initial and $ * $ is the final position. We have an initial configuration (dashed lines) where robots converge to the correct formation shape (solid lines) in every step.}
\end{figure}

\section{Conclusions}

In the current work, we have considered the formation control problem in which agents have heterogeneous sensing mechanism, distance-based or bearing-only. Using heterogeneous sensing rigidity framework for defining the desired shape and using the persistence notion for heterogeneous sensing formation, we study the use of gradient-based control law for achieving the desired formation with heterogeneous sensing information. The distributed heterogeneous control design and analysis start from a core team of agents that can maintain the formation and it is scaled up by adding an agent one-at-a-time to the network. Simulation results show the validity of the theoretical results.





\bibliographystyle{IEEEtran}
\bibliography{ref}


\begin{IEEEbiography}[{\includegraphics[width=1in,height=1.25in,clip,keepaspectratio]{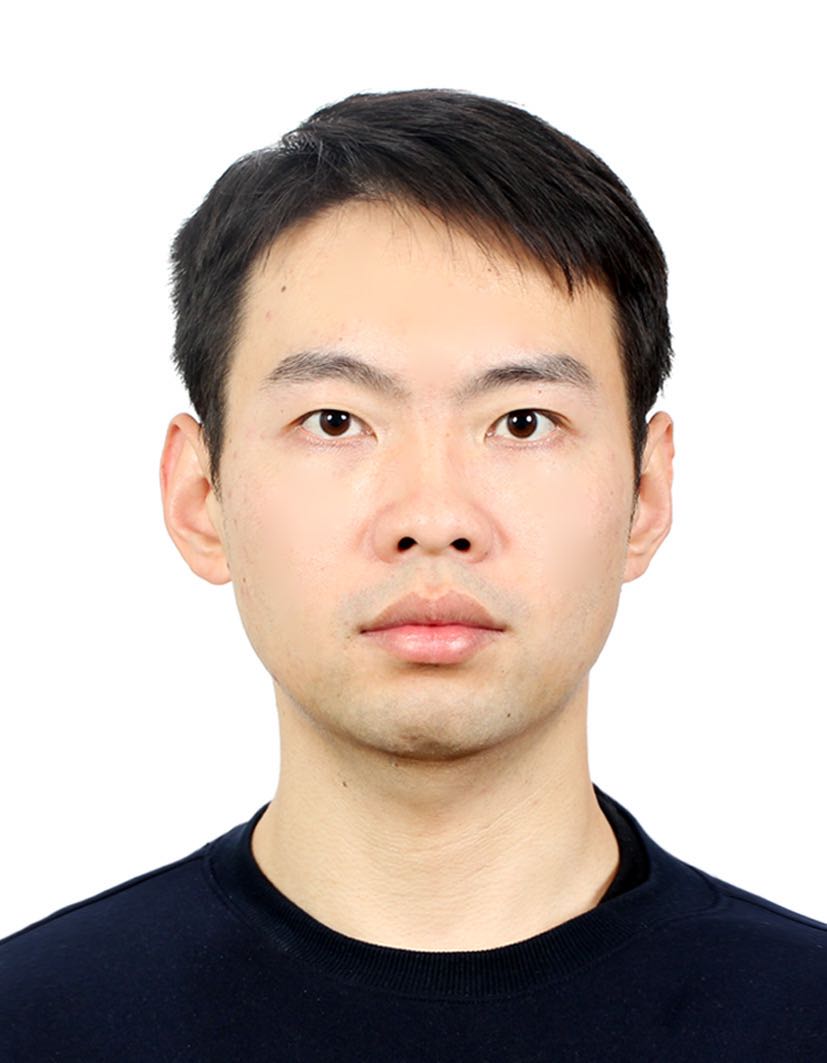}}]{Jin Chen}
received the B.Sc. degree in vehicle engineering from Shanghai Jiao Tong University, China, in 2020.
He is currently with the Faculty of Science and Engineering, University of Groningen, The Netherlands, working toward his Ph.D. degree.
His research interests include nonlinear control, formation control, and robotic systems.
\end{IEEEbiography}

\begin{IEEEbiography}[{\includegraphics[width=1in,height=1.25in,clip,keepaspectratio]{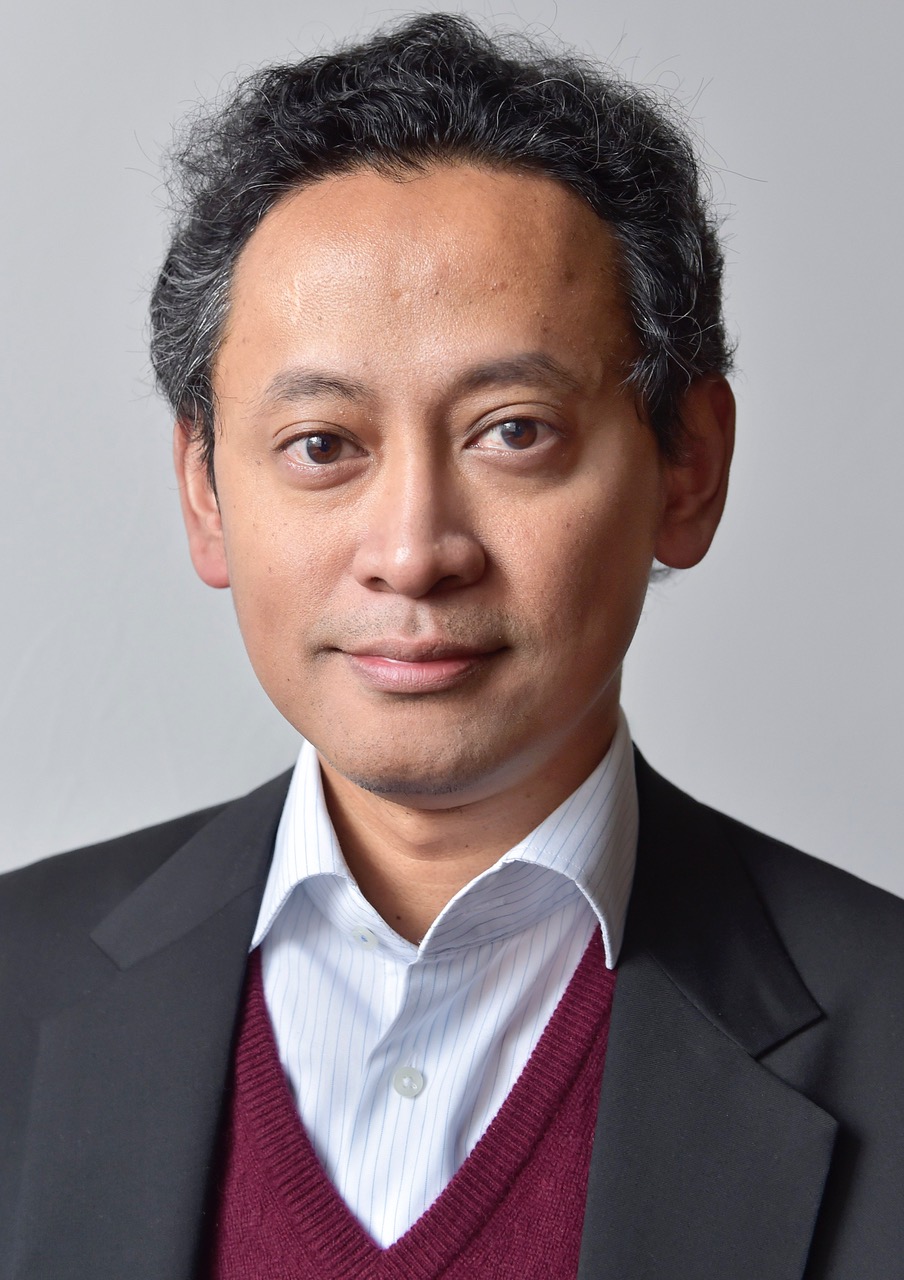}}]{Bayu Jayawardhana}
(SM’13) received the B.Sc. degree in electrical and electronics engineering from the Institut Teknologi Bandung, Bandung, Indonesia, in 2000, the M.Eng. degree in electrical and electronics engineering from the Nanyang Technological University, Singapore, in 2003, and the Ph.D. degree in electrical and electronics engineering from Imperial College London, London, U.K., in 2006. He is currently a professor of mechatronics and control of nonlinear systems in the Faculty of Science and Engineering, University of Groningen, The Netherlands. He was with Dept. Mathematical Sciences, Bath University, Bath, U.K., and with Manchester Interdisciplinary Biocentre, University of Manchester, Manchester, U.K. His research interests include the analysis of nonlinear systems, systems with hysteresis, mechatronics, systems and synthetic biology. Prof. Jayawardhana is a Subject Editor of the International Journal of Robust and Nonlinear Control and an Associate Editor of the European Journal of Control. 
\end{IEEEbiography}

\begin{IEEEbiography}[{\includegraphics[width=1in,height=1.25in,clip,keepaspectratio]{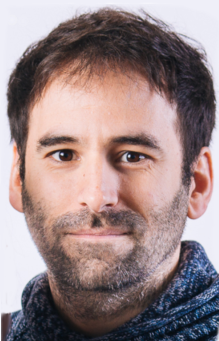}}]{H\'ector Garc\'ia de Marina} received the Ph.D. degree in systems and control from the University of Groningen, the Netherlands, in 2016.

He was a Postdoctoral Research Associate with the Ecole Nationale de l’Aviation Civile, Toulouse, France, and an Assistant Professor with the Unmanned Aerial Systems Center, University of Southern Denmark, Odense, Denmark. Since 2022, he is a Ramón y Cajal researcher with the Department of Computer Engineering, Automation and Robotics, and with CITIC, Universidad de Granada, Spain. He is an Associate Editor for IEEE Transactions on Robotics. His current research interests include multi-agent systems and the design of guidance navigation and control systems for autonomous vehicles.
\end{IEEEbiography}

\end{document}